\documentclass[12pt]{article}
\usepackage[utf8]{inputenc}
\usepackage{amssymb}
\usepackage{graphicx,epstopdf}
\usepackage{amsthm}
\usepackage{tikz}
\usepackage{subfig}
\usepackage{siunitx}
\RequirePackage[hyphens]{url}
\RequirePackage[colorlinks,citecolor=blue,urlcolor=blue]{hyperref}
\usepackage[numbers]{natbib}
\usepackage[lined,algonl,boxed]{algorithm2e}
\usepackage{paralist}
\usepackage{url}
\usepackage{amsmath}
\usepackage{amsfonts}
\usepackage{color}
\usepackage{afterpage}
\usepackage{lineno}
\usepackage{mathtools}
\usepackage[utf8]{inputenc}
\usepackage[english]{babel}
\usepackage{enumitem}
\usepackage[numbers]{natbib}
 
\newtheorem{theorem}{Theorem}[section]
\newtheorem{Remark}{Remark}[section]

\newtheorem{lemma}{Lemma}[section]

\newcommand{\ceil}[1]{\lceil {#1} \rceil}

\font\n=cmcsc10

\begin{document}
\title{Variational Bayes and truncation approximations for enriched Dirichlet process mixtures

}
\author{Somnath Bhadra, Michael J. Daniels}
\maketitle

Abstract:
A common impediment in conducting inference for Bayesian nonparametric models is either the need for complex MCMC algorithms and/or computational run-time for large datasets.  We propose solutions here for enriched Dirichlet process mixtures (EDPM).  We derive a variational Bayes estimator based on a previously developed truncation approximation for EDPMs.  The variational Bayes estimator can be used in two ways: 1) to develop a more efficient truncation approximation; 2) as good initial values for a blocked Gibbs sampler based on this more efficient truncation approximation or for a polya urn sampler. We derive the accuracy of this more efficient truncation approximation and demonstrate how this allows for simple implementation of a blocked Gibbs Sampler.  We confirm the validity of the approximations by simulations and illustrate on a real data set.


\section{Introduction}

Dirichlet Process mixture (DPM) and its enriched version, the enriched Dirichlet Process Mixture (EDPM) are popular approaches for Bayesian non-parametric regression and density estimation. MCMC is the standard method to obtain posterior samples and estimate parameters. However, determining good initial values for MCMC can be difficult, and the convergence time can greatly depend on the initial values. 

For the Variational Bayes (VB) approach, instead of generating samples from the posterior distributions, minimization of Kullback-Leibler (KL) divergence of the observed (empirical) distribution of the parameters with a pre-specified, simple family of distribution $\mathbb{Q}$ is used. This concept was first developed by \cite{PetersonAndAnderson1987}~ in a neural network setup. \cite{saul1996mean}~ and \cite{jordan1999introduction}~ extended the idea to a wider class of models. Around the same time, \cite{hinton1993keeping}~ and \cite{neal1998view}~ showed important connections between variational algorithms in the same neural network with the EM (Expectation-Maximization) algorithm, leading to variety of variational inference algorithms for other models as well. Like MCMC, variational inference methods have their roots in statistical physics, but VB is an optimization problem. More recent applications of VB can be found in \cite{Li_2022}~, \cite{Tommaso_2023}~ and \cite{zhong2025variationalinferencelatentvariable}~ .

The DPM approach for non-parametric regression, also called the joint approach, was first introduced by \cite{muller1996bayesian}~, and then further explored by many others including \cite{kang2009clusterwise}~, \cite{shahbaba2009nonlinear}~, \cite{hannah2011dirichlet}~, \cite{park2010bayesian}; and \cite{jara2011dppackage}. The DPM typically uses simple local linear regression models as building blocks and partitions the observed subjects into clusters, where within clusters, the simple regression model provides a good fit. Even though within clusters the model is parametric, globally, a wide range of complex distributions are allowed for the joint distribution, leading to a flexible model for both the regression function and the conditional density.

In this paper, we will derive a variational inference algorithm for EDP mixtures \cite{wade2011enriched}~ based on a truncation of its square-breaking representation \cite{burns2023truncation}, similar to the concept in \cite{ishwaran2001gibbs}, \cite{blei2006variational}~ and in \cite{blei2017variational}~. The algorithm involves the posterior distribution $p$ and a variational distribution, $q$. We define $q$ with variational parameters, and the algorithmic problem is to estimate these parameters so that $q$ approximates $p$. We will use the VB estimates to obtain a more efficient truncation approximation than in \cite{burns2023truncation}. 

As pointed out in \cite{wade2011enriched}~, VB is easier to compute, and often finds a density close to the target; whereas MCMC produces exact samples from the target density. Hence if we use MCMC with initial values from VB, we expect the convergence of the MCMC to be much faster, preferably with a shorter burn-in period. Also, since we control the choice of family of distribution for VB, it can be chosen such that the computations can be made easier and the approximation more accurate.

In this paper, in the section $\ref{sec:desc_EDPM}$, we review the EDPM and truncation approximations. In section $\ref{sec:VB}$ we quickly review VB and then extend that to an EDPM. Next in section $\ref{sec:trunc}$ we derive the accuracy of the truncation approximation and then use VB to find ``reasonable" truncation values based on these results. We also introduced a blocked Gibbs sampler (BGS). In section $\ref{sec:sim}$ we will demonstrate our method with some simulations. Next in section $\ref{sec:application}$ we will fit our method, and compare with different existing methods (EDPM with different truncations, and parametric AFT model). Finally in section $\ref{sec:conclusion}$ we will conclude our findings in this paper and will discuss some potential future works.

\section{The EDPM and a truncation approximation}\label{sec:desc_EDPM}

The EDPM in \cite{wade2011enriched}~ is specified as follows

\begin{equation*}\label{eq:setup}
\begin{split}
    & Y_i|x_i,\theta_i \sim F_y(\cdot|x_i,\theta_i)\\
    & x_i|\psi_i \sim F_x(\cdot|\psi_i)\\
    & \theta_i,\psi_i|P \sim P\\
    & P \sim EDP (\alpha_\theta,\alpha_{\psi|\theta},P_0=P_{0\theta}\times P_{0\psi_\theta})\\
    & \text{where}\\
    & P_\theta \sim DP(\alpha_\theta, P_{0\theta})\\
    & P_{\psi|\theta}(.|\theta) \sim DP(\alpha_{\psi|\theta}(\theta), P_{0\psi|\theta}(.|\theta)).\\
\end{split}   
\end{equation*}

This can be expressed using a square-breaking representation similar to that of a DP (\cite{sethuraman1994constructive},  

\begin{equation*}\label{Eq:Stick_break}
P=\sum_{k=1}^{\infty}\sum_{j=1}^\infty p_k^\theta p_{j|k}^\psi\delta_{\theta_k^*\psi_{j|k}^*}(\cdot),\\ 
\end{equation*}

\noindent where $p_1^\theta=V_1^\theta$ and $p_k^\theta=V_k^\theta\prod_{h=1}^{k-1}(1-V_h^\theta)$ and similarly for $V_{j|k}^\psi$; $V_k\stackrel{iid}{\sim}Beta (1,\alpha_\theta)$ and $V_{j|k}|\theta_k^*\stackrel{iid}{\sim}Beta (1,\alpha_{\psi|\theta}(\theta^*))$; $\theta_k^*\sim P_{0\theta}$, $\psi_j^*|\theta_k^*\sim P_{0\psi|\theta}(\theta_k^*)$.

The use of a truncation approximation of the infinite mixture in stick-breaking representation of DPM by a finite one was first introduced by \cite{ishwaran2001gibbs}. \cite{burns2023truncation} extended this concept to an EDPM setup. Assuming a fixed number of clusters $(N, M)$, they have shown 
$$P_{NM}(\cdot)=\sum_{k=1}^N\sum_{j=1}^M p_k^\theta p_{j|k}^\psi\delta_{\theta_k^*,\psi_{j|k}^*}.$$ $P_{NM}$ converges almost surely to an EDPM with base distribution $P_{0\theta}\times P_{0\psi|\theta}$ and precision parameters $\alpha^\theta$, $\alpha^{\psi|\theta}$.

\section{Variational Bayes for a truncated EDPM}\label{sec:VB}

 Variational Bayes (VB) frames the posterior density estimation problem as an optimization problem of minimizing distance between the observed posterior with that of a family of parameterized distributions, and then approximating it with the corresponding minimizer from the family.

\noindent Here we assume the posterior distributions of each of the components of the EDPM parameters $W=(V_k^\theta, V_{j|k}^\psi, \theta, \psi)$ come from a prefixed family of distributions, say $q_W$. For a fully factorized $q$, we assume 
\begin{equation}\label{eq:factorized}
    q_{_W}(W)=\prod_{i=1}^P q_{w_i}(W_i),
\end{equation}

\noindent where $P$ is the dimension of $W$, and $q_{w_i}$s are a parametric family of distributions.
 The optimal distributions of $W$ from the family $q_W$ is given as

\begin{equation*}
\begin{split}
f^*_{q_{_W}}={\arg\min}_{W\in q_{_W}}KL(f_n^*(.),f_{W}(.)) ,   
\end{split}
\end{equation*}

\noindent where $f_n^*(.)$ is the empirical distribution, $f_{W}(.)$ is a general form of pdf for the family $q_W$, and $KL(f,g)$ is the Kulback-Leibler divergence between densities $f$ and  $g$.

The minimization of KL divergence between these two densities can be computed by maximizing Evidence Lower BOund (ELBO), given as 
\begin{equation*}
    ELBO(q)=E_q(\log p(W,X,Y)|\theta)-E_q(\log q_W(W)),
\end{equation*}

\noindent where $p(W,X,Y)$ is the joint posterior distribution of the parameters $(W)$ and the data $(X,Y)$.

For the EDPM set of parameters,  $W=(V_k^\theta, V_{j|k}^\psi, \theta, \psi)$, we have $p_1^\theta=V_1^\theta$; $p_k^\theta=V_k^\theta\prod_{h=1}^{k-1}(1-V_h^\theta)$ and $p_{1|k}^\psi=V_{1|k}^\psi$; $p_{j|k}^\psi=V_{j|k}^\psi\prod_{h=1}^{j-1}(1-V_{h|k}^\psi)$; $\theta=(\theta_1,\ldots \theta_n)$; $\psi=(\psi_1,\ldots \psi_m)$. The first term in $ELBO(q)$ is

\begin{equation*}\label{eq:ELBO}
\begin{split}
    & E_q(\log p(W,X,Y|\theta^*,\psi^*,\alpha_\theta,\alpha_{\psi|\theta}))\\ & = E_q(\log p(V_k^\theta, V_{j|k}^\psi,\theta,\psi, X_1,\ldots X_n, Y_1, \ldots Y_n|\theta^*,\psi^*,\alpha_\theta,\alpha_{\psi|\theta}))\\
    & = \sum_{k=1}^N\Big[E_q(\log p(V_k^\theta|\alpha_\theta))+\sum_{j=1}^{M} E_q(\log p(V_{j|k}^\psi|\alpha_{\psi|\theta}))\Big] \\
    & + \sum_{i=1}^n E_q(\log p(Y_i|x_i,\theta_i))+\sum_iE_q(\log p(x_i|\psi_i))\\
    & +  \sum_{i=1}^n E_q(\log(p((\theta_i,\psi_i)|(\theta^*,\psi_{.|.}^*))))\\
    & + \sum_{k=1}^N\Big[E_q(\log p_{0\theta}(\theta_k^*))+\sum_{j=1}^{M} E_q(\log p_{0\psi|\theta}(\psi_{j|k}^*|\theta_k^*))\Big], \\
\end{split}
\end{equation*}

\noindent where $p(\theta_i,\psi_i|\theta^*,\psi_{.|.}^*)$ is a multinomial distribution with probabilities $p_k^{\theta^*}p_{j|k}^{\psi*}$ corresponding to pairs $(\theta_k^*,\psi_{j|k}^*)$.

Let $W_{-s}$ denotes the set of parameters from $W$ without the co-ordinate $W_s$. Assume $p(W_s|X,Y, W_{-s},\theta^*, \psi^*,\alpha_\theta, \alpha_{\theta|\psi})$ follows an exponential family, which holds for the EDPM, and is defined as 

\begin{equation*}
\begin{split}
    p(W_s|X,Y, W_{-s},\theta^*, \psi^*,\alpha_\theta, \alpha_{\theta|\psi})&=h_s(W_s)\exp\{g_s(W_{-s},X,Y|\theta^*, \psi^*,\alpha_\theta, \alpha_{\theta|\psi})^TW_s\\&-a_s(g_s(W_{-s},X,Y|\theta^*, \psi^*,\alpha_\theta, \alpha_{\theta|\psi}))\},
\end{split}    
\end{equation*}

\noindent where $g_s$ is the corresponding sufficient statistic of the conditional distribution of $W_s$ given all the other parameters.

\noindent To optimize the ELBO, we use a co-ordinate ascent algorithm with each step as follows;

\begin{equation}\label{eq:update}
\begin{split}
&\forall \hspace{0.1cm} 1\leq s\leq p,\\& \hspace{0.2cm}\nu_s=a_s^{\prime\prime}(\nu_s)^{-1}\frac{\partial}{\partial \nu_s}\Big[E_q(\log p(W_s|X,Y, W_{-s},\theta^*, \psi^*,\alpha_\theta, \alpha_{\theta|\psi})-\log h_s(W_s|\theta^*, \psi^*,\alpha_\theta, \alpha_{\theta|\psi}))\Big]      
\end{split}
\end{equation}
 
\noindent Then $(\ref{eq:update})$ can be simplified as

\begin{equation}\label{eq:update1}
\nu_s=E_q(g_s(W_{-s},X,Y|\theta^*, \psi^*,\alpha_\theta, \alpha_{\theta|\psi}))  
\end{equation}

\noindent and the approximated fully factorized family can be written as 
\begin{equation}\label{eq:prior}
    q_{\nu_s}(W_s)=h_s(W_s)\exp \{\nu_s^TW_s-a_s(\nu_s)\}.
\end{equation}

\noindent To compute the expectation in $(\ref{eq:update1})$, we need the joint distribution of $(W_{-s},X,Y|\theta^*, \psi^*,\alpha_\theta, \alpha_{\theta|\psi})$.

\noindent This can be computed as

\begin{equation*}\label{eq:huh?}
\begin{split}
    p(W_{-s},X,Y|\theta^*, \psi^*,\alpha_\theta, \alpha_{\theta|\psi}) & = \int p(W,(X,Y)|\theta^*, \psi^*,\alpha_\theta, \alpha_{\theta|\psi}) \partial W_s\\
    & = \int \prod_{j=1}^P p_{j}(W_j|\theta^*, \psi^*,\alpha_\theta, \alpha_{\theta|\psi}) \partial W_s\text{  \hspace{0.2cm}           by ($\ref{eq:factorized}$)}\\
    & = \prod_{j=1,j\neq m}^P p_{,j}(W_j|\theta^*, \psi^*,\alpha_\theta, \alpha_{\theta|\psi})
    \int p_{m}(W_s|\theta^*, \psi^*,\alpha_\theta, \alpha_{\theta|\psi}) \partial W_s\\
    & = \prod_{j=1,j\neq m}^P p_{j}(W_j|\theta^*, \psi^*,\alpha_\theta, \alpha_{\theta|\psi})\int h_s(W_s)\exp\Big\{\nu_s^T W_s-a_s(\nu_s)\Big\}\partial W_s\hspace{0.1cm}\\&[\text{by } (\ref{eq:prior}) \text{ and noting that } a_s(\nu_s) \text{ does not depend on } W_s]\\
    & = \Big[\prod_{j=1,j\neq m}^P p_{j}(W_j|\theta^*, \psi^*,\alpha_\theta, \alpha_{\theta|\psi})\Big]\\& \times \exp\{-a_s(\nu_s)\}\int h_s(W_s)\exp\{\nu_s^T W_s\}\partial W_s\\&[\text{Since $\int h_s(W_s)\exp\{\nu_s^T W_s\}\partial W_s$ is still of exponential family form}]
\end{split}
\end{equation*}

\noindent We will show this calculation for each set of parameters in $W$. In what follows we will use the below specifications for the exponential family distribution  of $Y$, $X$ and the pair $(\theta^*,\psi^*)$ which holds for the EDPM specification 

\begin{equation*}\label{eq:simpler}
\begin{split}
    & Y_i|x_i,\theta_i \sim f_y(\cdot|x_i,\theta_i)=h_Y(Y_i)g_Y^{(1)}(x_i)g_Y^{(2)}(\theta_i)\exp\{\sum_{r=1}^{k_Y}{\eta_{Y,r}^{(1)}(x_i)\eta_{Y,r}^{(2)}(\theta_i)T_{Y,r}(Y_i)}\}\\
    & x_i|\psi_i \sim f_x(\cdot|\psi_i)=h_X(x_i)g_X(\psi_i)\exp{\{\sum_{r=1}^{k_x}\eta_{X,r}(\psi_i)T_{X,r}(x_i)\}}\\
    & \theta^*_k \stackrel{iid}{\sim} P_{0\theta}=h_\theta(\theta^*_k)g_\theta(\theta)\exp{\{\sum_{r=1}^{k_\theta}\eta_{\theta,r}(\theta)T_{\theta,r}(\theta^*_k)\}}\\
    & \psi^*_{j|k}\stackrel{iid}{\sim}P_{0\psi|\theta}=h_\psi(\psi^*_j)g_\psi(\psi)\exp{\{\sum_{r=1}^{k_\psi}\eta_{\psi,r}(\psi)T_{\psi,r}(\psi^*_j)\}},\\
\end{split}    
\end{equation*}
\noindent where $k_Y,k_x,k_\theta,k_\psi$ are the dimensions of the corresponding sufficient statistics.

\noindent First take $W_s=\theta_i, \hspace{0.1cm}i \in \{1,2,\ldots n\}$ and derive $$p(W_{-s},X,Y|\theta^*,\psi^*,\alpha_\theta,\alpha_\psi)=p(\theta_1,\ldots \theta_{i-1},\theta_{i+1}, \ldots \theta_N,\psi_{.|\theta},V^\theta,V^\psi,X,Y|\theta^*, \psi^*,\alpha_\theta, \alpha_{\theta|\psi}).$$

\begin{equation}\label{eq:simpler_calculation_theta}
\begin{split}
    &p(W_{-s},X,Y|\theta^*,\psi^*,\alpha_\theta,\alpha_\psi)\\ & = \int (p(W,X,Y|\theta^*,\psi^*,\alpha_\theta,\alpha_\psi) \partial \theta_i\\
    & = \int \prod_{j=1}^P p_{\theta,j}(W_j) \partial \theta_i\\
    & \propto \prod_{j=1,W_j\neq \theta_i}^P p_{\theta,j}(W_j)
    \\&\times \sum_{\theta_i\in \theta^*} [h_Y(Y_i)g_Y^{(1)}(x_i)g_Y^{(2)}(\theta_i)\exp\{\sum_r{\eta_{Y,r}^{(1)}(x_i)\eta_{Y,r}^{(2)}(\theta_i)T_{Y,r}(Y_i)}\} \dfrac{n!}{\theta_i!(n-\sum_{j\neq i}^{k}\theta_j)!}(p^\theta_i/p^\theta_k)^{\theta_i}]\\
    & \hspace{8cm}\{\text{Since } \theta_i \text{ has a multinomial distribution}\}\\
    & = \prod_{j=1,W_j\neq \theta_i}^P p_{\theta,j}(W_j) \times \sum_{\theta_i\in \theta^*}\Big[ \dfrac{g_Y^{(2)}(\theta_i)c_Y^{\sum_r(\eta_{Y,r}^{(2)}(\theta_i))}}{\theta_i!(n-\sum_{j\neq i}^{k}\theta_j)!}(p^\theta_i/p^\theta_k)^{\theta_i}\Big]\\&\hspace{8cm}[c_Y=n! h_Y(Y_i)g_Y^{(1)}(x_i)\exp{\sum_r(\eta_{Y,r}^{(1)}(x_i)T_{Y,r}(Y_i))}]\\
\end{split}
\end{equation}

\noindent In $(\ref{eq:simpler_calculation_theta})$ each of the $p_{\theta,j}(W_j)$ are of exponential family form. Also, $$\sum_{\theta_i\in \theta^*} \dfrac{g_Y^{(2)}(\theta_i)c_Y^{\sum_r(\eta_{Y,r}^{(2)}(\theta_i))}}{\theta_i!(n-\sum_{j\neq i}^{k}\theta_j)!}(p^\theta_i/p^\theta_k)^{\theta_i}$$ is a summation of multinomial distribution terms, hence also in exponential family form. Hence, the distribution of $\theta_i$ conditioning on rest of the parameters in $(\ref{eq:simpler_calculation_theta})$ is an exponential family distribution.

\noindent Now take $W_s=\psi_i,\hspace{0.1cm}i \in \{1,2,\ldots m\}$. 

\begin{equation}\label{eq:simpler_calculation2}
\begin{split}
    & p(W_{-s},X,Y|\theta^*,\psi^*,\alpha_\theta,\alpha_\psi)\\ & = \int (p(W,X,Y|\theta^*,\psi^*,\alpha_\theta,\alpha_\psi) \partial \psi_i\\
    & = \int \prod_{j=1}^M p_{\theta,j}(W_j) \partial \psi_i\\
    & \propto \prod_{j=1,W_j\neq \psi_i}^P p_{\theta,j}(W_j) \times \sum_{\psi_i\in \psi^*} [h_X(x_i)g_X(\psi_i)\exp\{\sum_r{\eta_{X,r}(\psi_i)T_{X,r}(x_i)}\} \dfrac{n!}{\psi_i!(n-\sum_{l\neq i}^{j}\psi_l)!}(p^\psi_i/p^\psi_j)^{\psi_i}]\\
    & \hspace{8cm}\{\text{Since } \psi_i \text{ has a multinomial distribution}\}\\
    & = \prod_{j=1,W_j\neq \psi_i}^P p_{\theta,j}(W_j) \times \sum_{\psi_i\in \psi^*} \Big[\dfrac{g_X(\psi_i)c_X^{\sum_r(\eta_{X,r}(\psi_i))}}{\psi_i!(n-\sum_{l\neq i}^{j}\psi_l)!}(p^\psi_i/p^\psi_j)^{\psi_i}\Big]\hspace{0.5cm}[c_X=n! h_X(x_i)\exp{\sum_r(T_{X,r}(x_i))}].\\
\end{split}
\end{equation}

\noindent In $(\ref{eq:simpler_calculation2})$ since $$\sum_{\psi_i\in \psi^*} \dfrac{g_X(\psi_i)c_X^{\sum_r(\eta_{X,r}(\psi_i))}}{\psi_i!(n-\sum_{l \neq i}^{j}\psi_l)!}(p^\psi_i/p^\psi_j)^{\psi_i}$$ is a summation of multinomial distribution terms, hence also in exponential family form. Hence distribution of $\psi_i$ conditioning on rest of the parameters is an exponential family distribution.

\vspace{0.5cm}
\noindent Now consider $W_s=p^\theta_i, \hspace{0.1cm}i \in \{1,2,\ldots n\}$. First note that the set of equations $p^\theta_k=V^\theta_k\prod_{j=1}^{k-1}(1-V^\theta_j)$ implies $V^\theta_k=\dfrac{p^\theta_k}{1-\sum_{j=1}^{k-1}p^\theta_j}$

Hence,

\begin{align}\label{eq:simpler_calculation_p}
& p(W_{-s}, X, Y | \theta^*, \psi^*, \alpha_0, \alpha_\psi)\notag\\
&= \int p(W,X,Y |\theta^*,\psi^*,\alpha_0,\alpha_\psi)\, d p_i^\theta \notag\\
&= \int \prod_{j=1}^P p_{\theta,j}(W_j)\, d p_i^\theta \notag\\
&\propto \prod_{\substack{j=1 \\ W_j \neq p_i^\theta}}^P p_{\theta,j}(W_j)
   \times \int (1 - p_i^\theta)^{\alpha_0-1}
   \left( \frac{p_i^\theta}{p_k^\theta} \right)^{\theta_i} d p_i^\theta.
\end{align}

\noindent In $(\ref{eq:simpler_calculation_p})$, each of the $p_{\theta,j}(W_j)$ are of exponential family form. Also, since \\ $$\int (1-p^\theta_i)^{\alpha_\theta-1}\Big(\dfrac{p^\theta_i}{p^\theta_k}\Big)^{\theta_i} \partial p^\theta_i$$ is the integral of a $Beta (\theta_i+1,\alpha_\theta)$ distribution (an exponential family member), it is in exponential family form. Hence the distribution of $p^\theta_i$ conditioning on rest of the parameters is an exponential family distribution.

\noindent Finally, consider $W_s=p^\psi_{j|i}$, and note that $V^\psi_{k|l}=\dfrac{p^\psi_{k|l}}{1-\sum_{j=1}^{k-1}p^\psi_{j|l}}$. Hence,

\begin{align}\label{eq:simpler_calculation5}
p(W_{-s}, X, Y | \theta^*, \psi^*, \alpha_0, \alpha_\psi)
&= \int p(W, X, Y | \theta^*, \psi^*, \alpha_0, \alpha_\psi) \, d p^{\psi}_{j|i} \notag \\
&= \int \prod_{j=1}^P p_{\theta,j}(W_j) \, d p^{\psi}_{j|i} \notag \\
&\propto \prod_{\substack{j=1 \\ W_j \neq p^{\psi}_{j|i}}}^P p_{\theta,j}(W_j)
    \times \int (1 - p^{\psi}_{j|i})^{\alpha_\psi - 1}
    \left( \dfrac{p^{\psi}_{j|i}}{p^{\psi}_{k|i}} \right)^{\psi_j}
    d p^{\psi}_{j|i}.
\end{align}

\noindent In $(\ref{eq:simpler_calculation5})$, each of the $p_{\theta,j}(W_j)$ are of exponential family form. Furthermore, since $$\int (1-p^\psi_{j|i})^{\alpha_\psi-1}\Big(\dfrac{p^\psi_{j|i}}{p^\psi_{k|i}}\Big)^{\psi_j} \partial p^\psi_{j|i}$$, is an  integral of a $Beta(\psi_j+1,\alpha_\psi)$ distribution (an exponential family member), it is in exponential family form. Hence the distribution of $p^\psi_{j|i}$ conditioning on rest of the parameters is an exponential family distribution.

\section{A New Truncation Approximation}\label{sec:trunc}

The validity of the approximation of stick-breaking method for DPM is shown in \cite{ishwaran2001gibbs}. Using a similar idea from \cite{burns2023truncation}~ and reviewed in Section $\ref{sec:desc_EDPM}$, we will extend to the case with $M_k\neq M$ which will be more efficient as we often do not need same number of $X$ clusters for each $Y$ cluster.
\subsection{Truncation}

Define 
\begin{equation}\label{eq:truncEDP}
P_N^{\mathbb{M}}(\cdot)=\sum_{i=1}^N\sum_{j=1}^{M_k}p_k^\theta p_{j|k}^\psi\delta_{\theta_k^*\psi_{j|k}^*}(\cdot)   
\end{equation}

and 
\begin{equation}\label{eq:trueEDP}
 P_\infty(\cdot)= EDP(\alpha^\theta,\alpha^{\psi|\theta},P_0)= \sum_{i=1}^\infty \sum_{j=1}^\infty \pi_k^\theta \pi_{j|k}^\psi\delta_{\theta_k^*\psi_{j|k}^*}(\cdot),  
\end{equation}

\noindent where $\pi_1^\theta=Z_1^\theta$; $\pi_k^\theta=Z_k^\theta\prod_{h=1}^{k-1}(1-Z_h^\theta)$; similarly for $\pi_{j|k}^{\psi}$; $Z_k^\theta \sim Beta(1,\alpha_\theta),\hspace{0.1cm}k=2,3,\ldots$; $\theta^*_k\stackrel{iid}{\sim} P_{0\theta}, \hspace{0.1cm} k=1,2,\ldots$; and for each $k$, $Z_{j|k}^{\psi}|\theta^*_k\stackrel{iid}{\sim}Beta(1,\alpha^{\psi|\theta}(\theta_k^*))$ and $\psi^*_{j|k}|\theta^*_k\stackrel{iid}{\sim}P_{0\psi|\theta}(.|\theta_k^*), \hspace{0.1cm}j=1,2,\ldots$.

\noindent Denote by $D(\mathbb{P}_1,\mathbb{P}_2)$ the total variation distance between two probability measures $\mathbb{P}_1$ and $\mathbb{P}_2$. Then, $P_N^{\mathbb{M}}(\cdot)$ is a ``good" approximation for $P_\infty$ if the difference can be made as small as possible almost surely. 

\begin{theorem}
For $M_k\to\infty\hspace{0.1cm}\forall \hspace{0.1cm}1\leq k\leq N$ and $N\to \infty$

\begin{equation*}
    D(P_\infty,P_N^{\mathbb{M}}) \to 0 \hspace{0.5cm} w.p. 1
\end{equation*}
    
\end{theorem}

\begin{proof}

\begin{equation*}\label{eq:approxworks}
\begin{split}
  D(P_\infty,P_N^{\mathbb{M}}) & = \sup_A\Big|\Big (\sum_{k=N}^\infty\sum_{j=1}^\infty \pi_k^\theta \pi_{j|k}^\psi\delta_{\theta_k^*\psi_{j|k}^*}(A)-\sum_{j=1}^{M_N} p_N^\theta p_{j|N}^\psi\delta_{\theta_N^*\psi_{j|N}^*}\Big)\\
  & +\sum_{k=1}^{N-1}\Big[\sum_{j=M_k}^\infty \pi_k^\theta \pi_{j|k}^\psi\delta_{\theta_k^*\psi_{j|k}^*}(A)-p_k^\theta p_{M_k|k}^\psi\delta_{\theta_k^*\psi_{M_k|k}^*}(A)\Big]\Big|\hspace{0.5cm}[\text{from $(\ref{eq:truncEDP})$ and $(\ref{eq:trueEDP})$}]\\
  & \leq \sup_A\Big|\sum_{k=N}^\infty\sum_{j=1}^\infty \pi_k^\theta \pi_{j|k}^\psi\delta_{\theta_k^*\psi_{j|k}^*}(A)\Big|+\sup_A\Big|\sum_{j=1}^{M_N} p_N^\theta p_{j|N}^\psi\delta_{\theta_N^*\psi_{j|N}^*}\Big|\\
  & +\sum_{k=1}^{N-1}\Big[\sup_A\Big|\sum_{j=M_k}^\infty \pi_k^\theta \pi_{j|k}^\psi\delta_{\theta_k^*\psi_{j|k}^*}(A)\Big|+\sup_A\Big|p_k^\theta p_{M_k|k}^\psi\delta_{\theta_k^*\psi_{M_k|k}^*}(A)\Big|\Big]\\
  & = \sum_{k=N}^\infty \pi_k^\theta\sum_{j=1}^\infty \pi_{j|k}^\psi+p_N^\theta\sum_{j=1}^{M_N}p_{j|N}^\psi+\sum_{k=1}^{N-1}\Big[\pi_k^\theta\sum_{j=M_k}^\infty\pi_{j|k}^\psi+p_k^\theta p_{M_k|k}^\psi\Big]\\
  & = 1-\sum_{k=1}^{N-1}\pi_k^\theta + \Big(1-\sum_{k=1}^{N-1}\pi_k^\theta\Big)\sum_{j=1}^{M_N}p_{j|N}^\psi+\sum_{k=1}^{N-1}\Big[\pi_k^\theta\Big(1-\sum_{j=1}^{M_k}\pi_{j|k}^\psi\Big)+p_k^\theta\Big(1-\sum_{j=1}^{M_k}\pi_{j|k}^\psi\Big)\Big]\\
  & \stackrel{p}{\to} 2\Big(1-\sum_{k=1}^{N-1}\pi_k^\theta\Big) \to 0.\\
\end{split}    
\end{equation*}

\noindent Hence the square-breaking approximation with varying cluster sizes converges to its infinite mixture version.
\end{proof}

Now, using the square-breaking truncation approximation, the joint marginal density of $y$ and $x$ can be expressed as 

\begin{equation}\label{eq:defineM}
 m_N^{\mathbb{M}}(y,x)= \int\Big[ \prod_{i=1}^n\int f(y_i|x_i,\theta_i)f(x_i|\psi_i)dP(\theta_i,\psi_i)\Big]dP_n^{\mathbb{M}}(P)   
\end{equation}

\noindent Here we try to find a similar $L_1$ bound to the difference of $m_N^{\mathbb{M}}(\cdot)$ and $m_\infty(\cdot)$ as in \cite{burns2023truncation}, where $m_\infty(\cdot)$ denotes the marginal for true EDPM model. Denote the distribution under $P_n^{\mathbb{M}}$ as $\pi_n^{\mathbb{M}}$ and $P_\infty$ as $\pi_\infty$. From ($\ref{eq:defineM}$), we have the following theorem.

\begin{theorem}\label{thm:inequality}
\begin{equation}\label{eq:ineqality_thm}
\begin{split}
\int \Big|m_N^{\mathbb{M}}(y,x)-m_\infty(y,x)\Big|d(y,x) 
& \leq 4n \Big[\exp\Big(-\frac{N-1}{\alpha^\theta}\Big)\\
& +\exp\Big(-\min_k\Big\{\frac{M_k-1}{\alpha_k^{\psi|\theta}}\Big\}\Big)\Big)\Big(1-\exp{\Big(-\frac{N-1}{\alpha^\theta}\Big)}\Big)\Big]\\
\end{split}
\end{equation}

\noindent as $M_k\to \infty$ for each $k$ and $N\to \infty$. The RHS of the inequality can be made arbitrarily small.
\end{theorem}

\begin{proof}

\begin{equation}\label{eq:relate to TV}
\begin{split}
\int \Big|m_N^{\mathbb{M}}(y,x)-m_\infty(y,x)\Big|d(y,x) & = \int \Big|\int\prod_{i=1}^n f(y_i|x_i,\theta_i)f(x_i|\psi_i)(d\pi_n^{\mathbb{M}}(\theta,\psi)-d\pi_\infty(\theta,\psi))\Big| d(y,x)\\ 
& \leq \iint \prod_{i=1}^n f(y_i|x_i,\theta_i)f(x_i|\psi_i) \Big|(d\pi_n^{\mathbb{M}}(\theta,\psi)-d\pi_\infty(\theta,\psi))\Big| d(y,x)\\
& = \int\Big[\prod_{i=1}^n f(y_i|x_i,\theta_i)f(x_i|\psi_i) d(y,x)\int\Big] \Big|(d\pi_n^{\mathbb{M}}(\theta,\psi)-d\pi_\infty(\theta,\psi))\Big|\\
& = \int \Big|(d\pi_n^{\mathbb{M}}(\theta,\psi)-d\pi_\infty(\theta,\psi))\Big|\\
& = 2 D(\pi_n^{\mathbb{M}},\pi_\infty),\\
\end{split}
\end{equation}

\noindent where the last line follows from the usual definition of total variation distance.

Writing any $(\theta_i,\psi_i)$ observed pair as $(\theta_{K_i},\psi_{J_i|K_i})$, for $K_i<N$, $J_i<M_{K_i}$, and denoting $$C=\{K_i<N,J_i<M_{K_i},i=1,\ldots n\},$$ we can write;

\begin{equation}\label{eq:boundary}
\begin{split}
D(\pi_n^{\mathbb{M}},\pi_\infty) & = \sup_A \Big|\pi_n^{\mathbb{M}}(A)-\pi_\infty(A)\Big|\\
& = \sup_A \Big|\pi_n^{\mathbb{M}}(A\cap C^c)-\pi_\infty(A\cap C^c)\Big|\\
& \leq \sup_A \pi_n^{\mathbb{M}}(A\cap C^c) + \sup_A \pi_\infty(A\cap C^c)\\
& = 2(1-\pi_n^{\mathbb{M}}(C))\\
& = 2\Big[1-\mathbb{E}\Big(\sum_{k=1}^{N-1}p_k^\theta\sum_{j=1}^{M_k-1}p_{j|k}^\psi\Big)^n\Big].\\
\end{split}    
\end{equation}

Thus the difference between $P_N^\mathbb{M}$ and $P_\infty$ in $L_1$ is bounded by 
$2\Big[1-\mathbb{E}\Big(\sum_{k=1}^{N-1}p_k^\theta\sum_{j=1}^{M_k-1}p_{j|k}^\psi\Big)^n\Big]$.

Now note
\begin{equation}\label{eq:approx_expec}
\begin{split}
    \mathbb{E}\Big[\Big(\sum_{k=1}^{N-1}p_k^\theta\sum_{j=1}^{M_k-1}p_{j|k}^\psi\Big)\Big]^n & = \mathbb{E}\Big[\Big(\sum_{k=1}^{N-1}p_k^\theta(1-p_{M_k|k}^\psi)\Big)^n\Big]\\
    & = \mathbb{E}\Big[\Big(1-p_N^\theta-\sum_{k=1}^{N-1}p_k^\theta p_{M_k|k}^\psi\Big)^n\Big]\\
    & \approx \mathbb{E}\Big[1-n\Big(p_N^\theta + \sum_{k=1}^{N-1}p_k^\theta p_{M_k|k}^\psi\Big)\Big].\\
\end{split}
\end{equation}

And noting the form of $p_N^\theta$, we have 

\begin{equation}\label{eq:approx_k}
\begin{split}
    \mathbb{E}(p_N^\theta)& =\mathbb{E}\Big[\prod_{k=1}^{N-1}(1-V_k^\theta)\Big]\\
    & = \mathbb{E}\Big[\exp{\Big(\sum_{k=1}^{N-1}-\frac{1}{\alpha^\theta}E_k^\theta\Big)}\Big] \hspace{0.5cm}(\text{where }E_k^\theta \stackrel{iid}{\sim}\exp{(1)})\\
    & \approx \exp{\Big(-\frac{N-1}{\alpha^\theta}\Big)}.\\
\end{split}    
\end{equation}

Similarly,

\begin{equation}\label{eq:approx_each j}
\begin{split}
    \mathbb{E}(p_k^\theta p_{M_k|k}^\psi) & = \mathbb{E}\mathbb{E}\Big[p_k^\theta p_{M_k|k}^\psi\Big|p_1^\theta,\ldots p_{N-1}^\theta,\theta_1^*,\ldots \theta_N^*\Big]\\
    & = \mathbb{E}\Big[p_k^\theta\mathbb{E}\Big(p_{M_k|k}^\psi\Big|\theta_k^*\Big)\Big]\\
    & = \mathbb{E}\Big[p_k^\theta\mathbb{E}\Big(\prod_{j=1}^{M_k-1}(1-V_{j|k}^\psi(\theta_k^*))\Big)\Big]\\
    & = \mathbb{E}\Big[p_k^\theta \times \exp{-\frac{1}{\alpha^{\psi|\theta}}\sum_{j=1}^{M_k-1}E_{j|k}^\psi(\theta^*)}\Big] \hspace{0.5cm}(\text{where }E_{j|k}^\psi(\theta^*) \stackrel{iid}{\sim}\exp{(1)})\\
    & \approx \mathbb{E}\Big(p_k^\theta \times \exp{\Big(-\frac{M_k-1}{\alpha_k^{\psi|\theta}}\Big)}\Big) \\
    & = \exp{\Big(-\frac{M_k-1}{\alpha_k^{\psi|\theta}}\Big)}\mathbb{E}(p_k^\theta).
\end{split}   
\end{equation}

\noindent Thus from $(\ref{eq:boundary})$, $(\ref{eq:approx_expec})$, $(\ref{eq:approx_k})$ and $(\ref{eq:approx_each j})$ we have 

\begin{equation}\label{eq:final_bound}
\begin{split}
    D(\pi_n^{\mathbb{M}},\pi_\infty) & \leq 2n\Big[\exp\Big(-\frac{N-1}{\alpha^\theta}\Big)+\sum_{k=1}^{N-1}\exp\Big(-\frac{M_k-1}{\alpha_k^{\psi|\theta}}\Big)\mathbb{E}(p_k^\theta)\Big)\Big]\\
    & \leq 2n \Big[\exp\Big(-\frac{N-1}{\alpha^\theta}\Big)+\exp\Big(-\min_k\Big\{\frac{M_k-1}{\alpha_k^{\psi|\theta}}\Big\}\Big)\sum_{k=1}^{N-1}\mathbb{E}(p_k^\theta)\Big)\Big]\\
    & = 2n \Big[\exp\Big(-\frac{N-1}{\alpha^\theta}\Big)+\exp\Big(-\min_k\Big\{\frac{M_k-1}{\alpha_k^{\psi|\theta}}\Big\}\Big)\mathbb{E}(1-p_N^\theta)\Big)\Big]\\
    & \approx 2n \Big[\exp\Big(-\frac{N-1}{\alpha^\theta}\Big)+\exp\Big(-\min_k\Big\{\frac{M_k-1}{\alpha_k^{\psi|\theta}}\Big\}\Big)\Big)\Big(1-\exp{\Big(-\frac{N-1}{\alpha^\theta}\Big)}\Big)\Big].\\
\end{split}    
\end{equation}

\noindent Combining $(\ref{eq:relate to TV})$ and $(\ref{eq:final_bound})$, Theorem $\ref{thm:inequality}$ follows.

\end{proof}

\noindent We can also use Theorem $\ref{thm:inequality}$ to find a $\mathcal{L}_1$ bound on the posterior distribution of the clustering. 

\begin{lemma}\label{lem:L1}
\begin{equation*}
\begin{split}
&\int_{\mathbb{R}^n\times \mathbb{R}^n}\Big[\sum_{K\in K_\infty}\sum_{J\in J_\infty}\Big|\pi_{N}^\mathbb{M}(K,J|y,X)-\pi_{\infty}(K,J|y,x)\Big|\Big]m_\infty(y,x)d(y,x)\\
& = O\Big(n \Big[\exp\Big(-\frac{N-1}{\alpha^\theta}\Big)+\exp\Big(-\min_k\Big\{\frac{M_k-1}{\alpha_k^{\psi|\theta}}\Big\}\Big)\Big)\Big(1-\exp{\Big(-\frac{N-1}{\alpha^\theta}\Big)}\Big)\Big]\Big),           
\end{split}
\end{equation*}
\end{lemma}

\noindent where $K$ and $J$ are the vector of the clustering variables and $K_r=\{1,2,\ldots r\}^n$ and $J_q=\{1,2,\ldots q\}^n$ for any $r,q\in \mathbb{Z}^+$.

\begin{proof}

 Before proving the lemma, note that for $K\in K_\infty-K_N$, $\pi_{N}^\mathbb{M}(.)=0$; $J\in J_\infty-J_{M_K}$, $\pi_{N}^{M}(.)=0$.
 
 To prove the lemma, we first note the following equality.

\begin{equation}\label{eq:split for lemma2}
\begin{split}
& \sum_{K\in K_\infty}\sum_{J\in J_\infty}\Big|\pi_{N}^\mathbb{M}(K,J|y,X)-\pi_{\infty}(K,J|y,x)\Big|\\
& = \sum_{K\in K_N}\sum_{J\in J_{M_K}}\Big|\pi_{N}^\mathbb{M}(K,J|y,X)-\pi_{\infty}(K,J|y,x)\Big|+\sum_{K\in K_N}\sum_{J\in J_\infty-J_{M_K}}\pi_{\infty}(K,J|y,x)+\\&\sum_{K\in K_\infty-K_N}\sum_{J\in J_{M_K}}\pi_{\infty}(K,J|y,x).\\
\end{split}    
\end{equation}

\noindent  Since
\begin{equation*}
\begin{split}
\pi_{N}^\mathbb{M}(K,J|y,X)=\dfrac{P_N^\mathbb{M}(K,J)m_N^\mathbb{M}(y,x|K,J)}{m_N^\mathbb{M}(y,x)},
\end{split}   
\end{equation*}

\noindent where $P_N^\mathbb{M}(K,J)$ is the prior for $(K,J)$ under $P_N^\mathbb{M}$ and $$m_N^\mathbb{M}(y,x|K,J)=\prod_{k\in K^*,j_k\in J_K^*}\Big[\int\Big(\prod_{i:\{K_i=k,J_{K_i}=j_k\}}f(y_i|x_i,\theta)f(x_i|\psi)\Big)dP_0(\theta,\psi)\Big].$$

\noindent where $K^*$ and $J^*$ are the sets of unique $K$ and $J$ values.

\noindent Then 
\begin{equation}\label{eq:split_pi-diff}
\begin{split}
& \sum_{K\in K_N}\sum_{J\in J_{M_K}}\Big|\pi_{N}^\mathbb{M}(K,J|y,X)-\pi_{\infty}(K,J|y,x)\Big|\\
& = \sum_{K\in K_N}\sum_{J\in J_{M_K}}\Big|\dfrac{P_N^\mathbb{M}(K,J)m_N^\mathbb{M}(y,x|K,J)}{m_N^\mathbb{M}(y,x)}-\dfrac{P_\infty(K,J)m_\infty(y,x|K,J)}{m_\infty(y,x)}\Big|\\
& = \sum_{K\in K_N}\sum_{J\in J_{M_K}}\Big|\dfrac{P_N^\mathbb{M}(K,J)m_N^\mathbb{M}(y,x|K,J)}{m_N^\mathbb{M}(y,x)}-\dfrac{\{P_\infty(K,J)-P_N^\mathbb{M}(K,J)\}m_N^\mathbb{M}(y,x|K,J)}{m_\infty(y,x)}\\&-\dfrac{P_N^\mathbb{M}(K,J)m_N^\mathbb{M}(y,x|K,J)}{m_\infty(y,x)}\Big|\\
& \leq \sum_{K\in K_N}\sum_{J\in J_{M_K}}\dfrac{P_N^\mathbb{M}(K,J)m_N^\mathbb{M}(y,x|K,J)}{m_N^\mathbb{M}(y,x)}\Big|1-\dfrac{m_N^\mathbb{M}(y,x)}{m_\infty(y,x)}\Big|\\&+\sum_{K\in K_N}\sum_{J\in J_{M_K}}\Big|P_\infty(K,J)-P_N^\mathbb{M}(K,J)\Big|\dfrac{m_N^\mathbb{M}(y,x|K,J)}{m_\infty(y,x)}\\
& = \Big|1-\dfrac{m_N^\mathbb{M}(y,x)}{m_\infty(y,x)}\Big|\sum_{K\in K_N}\sum_{J\in J_{M_K}}\pi_N^\mathbb{M}(K,J|y,x)+\sum_{K\in K_N-K_{N-1}}\sum_{J\in J_{M_K}}\Big|P_\infty(K,J)-P_N^\mathbb{M}(K,J)\Big|\dfrac{m_N^\mathbb{M}(y,x|K,J)}{m_\infty(y,x)}\\&+\sum_{K\in K_{N-1}}\sum_{J\in J_{M_K}-J_{M_K-1}}\Big|P_\infty(K,J)-P_N^\mathbb{M}(K,J)\Big|\dfrac{m_N^\mathbb{M}(y,x|K,J)}{m_\infty(y,x)}\\
& \leq \Big|1-\dfrac{m_N^\mathbb{M}(y,x)}{m_\infty(y,x)}\Big|+\sum_{K\in K_N-K_{N-1}}\sum_{J\in J_{M_K}}\Big|P_\infty(K,J)-P_N^\mathbb{M}(K,J)\Big|\dfrac{m_N^\mathbb{M}(y,x|K,J)}{m_\infty(y,x)}\\&+\sum_{K\in K_{N-1}}\sum_{J\in J_{M_K}-J_{M_K-1}}\Big|P_\infty(K,J)-P_N^\mathbb{M}(K,J)\Big|\dfrac{m_N^\mathbb{M}(y,x|K,J)}{m_\infty(y,x)}.\\
\end{split}    
\end{equation}

\noindent where the last line follows since $K\in K_{N-1}$ and for each $K$, for $J\in J_{M_K-1}$, $P_\infty = P_N^\mathbb{M}$.

From $(\ref{eq:split_pi-diff})$,

\begin{equation}\label{eq:int_first-part}
\begin{split}
& \int_{\mathbb{R}^n\times \mathbb{R}^n}\sum_{K\in K_N}\sum_{J\in J_{M_K}}\Big|\pi_{N}^\mathbb{M}(K,J|y,X)-\pi_{\infty}(K,J|y,x)\Big|m_\infty(y,x) d(y,x)\\
& \leq \int\Big|m_\infty(y,x)-m_N^\mathbb{M}(y,x)\Big|d(y,x)+\int \sum_{K\in K_N-K_{N-1}}\sum_{J\in J_{M_K}}\Big|P_\infty(K,J)-P_N^\mathbb{M}(K,J)\Big|m_N^\mathbb{M}(y,x|K,J)d(y,x)\\& + \int \sum_{K\in K_{N-1}}\sum_{J\in J_{M_K}-J_{M_K-1}}\Big|P_\infty(K,J)-P_N^\mathbb{M}(K,J)\Big|m_N^\mathbb{M}(y,x|K,J)d(y,x)\\
& \leq \int\Big|m_\infty(y,x)-m_N^\mathbb{M}(y,x)\Big|d(y,x)+2 \int \sum_{K\in K_\infty}\sum_{J\in J_\infty}\Big|P_\infty(K,J)-P_N^\mathbb{M}(K,J)\Big|m_N^\mathbb{M}(y,x|K,J)d(y,x)\\
& = \int\Big|m_\infty(y,x)-m_N^\mathbb{M}(y,x)\Big|d(y,x)+4 D(P_\infty,P_N^\mathbb{M}).\\
\end{split}    
\end{equation}

\noindent Thus, by Theorem $\ref{thm:inequality}$, 
\begin{equation}\label{eq:rate}
\begin{split}
& \int_{\mathbb{R}^n\times \mathbb{R}^n}\sum_{K\in K_N}\sum_{J\in J_{M_K}}\Big|\pi_{N}^\mathbb{M}(K,J|y,X)-\pi_{\infty}(K,J|y,x)\Big|m_\infty(y,x) d(y,x)\\ & = O\Big(n \Big[\exp\Big(-\frac{N-1}{\alpha^\theta}\Big)+\exp\Big(-\min_k\Big\{\frac{M_k-1}{\alpha_k^{\psi|\theta}}\Big\}\Big)\Big)\Big(1-\exp{\Big(-\frac{N-1}{\alpha^\theta}\Big)}\Big)\Big]\Big)\\    
\end{split}    
\end{equation}
 
\noindent Using a similar argument for the other two terms in $(\ref{eq:split for lemma2})$, we see that the integral of both of the terms with respect to $m_\infty$ can be written as the RHS of $(\ref{eq:int_first-part})$.

Thus, together with equation $(\ref{eq:rate})$, Lemma $\ref{lem:L1}$ follows.

\end{proof}

\subsection{How to choose $N$ and $M_k$}

We propose an approach to choose the truncation bounds, $N$ and $M_k$'s to ensure the truncation approximation is sufficiently accurate based on Theorem $\ref{thm:inequality}$.

Given an overall maximum error value and an error value corresponding to only the $\theta$ clusters $\epsilon_\theta$, we use the first part of the RHS of the Theorem $\ref{thm:inequality}$ result to determine the choice of $N$ similar to \cite{ishwaran2001gibbs}. For the choice of $M_k$, based on each $\alpha_k^{\psi|\theta}$, we first find the corresponding $\tilde{M_k}$ value that would satisfy the desired bound. Let $k^*$ be defined as the $k$ value corresponding to the $\min_k\frac{\tilde{M_k}-1}{\alpha_k^{\psi|\theta}}$. We then choose the remaining $M_k$'s such that they are close to the $\dfrac{M_{k^*}-1}{\alpha_{k^*}^{\psi|\theta}}$ value.

\noindent In particular the equations we use to compute the integer values of $N$ and $M_k$'s are

\begin{equation}\label{eq:cluster_values}
\begin{split}
& n\Big[\exp\Big(-\dfrac{N-1}{\alpha^\theta}\Big)\Big]=\epsilon_\theta\\
& n\exp\Big(-\Big\{\frac{\tilde{M_k}-1}{\alpha_k^{\psi|\theta}}\Big\}\Big)\Big)\Big(1-\dfrac{\epsilon_\theta}{n}\Big)=\epsilon-\epsilon_\theta\hspace{1cm}\text{for each }k\\
& k^* = \arg\min_{1\leq k\leq N}\frac{\tilde{M_k}-1}{\alpha_k^{\psi|\theta}}\\
& M_{k^*} = \tilde{M}_{k^*}\\
& M_k=\ceil{\Big\{1+\Big(\dfrac{\alpha_k^{\psi|\theta}}{\alpha_{k^*}^{\psi|\theta}}\Big)(M_{k^*}-1)\Big\}}\hspace{1cm}\forall k\neq k^*.
\end{split}    
\end{equation}

\noindent The choice of $N$ and $\{M_k\}_k$ in $(\ref{eq:cluster_values})$ is driven by the upper bound in Theorem $\ref{thm:inequality}$. We first choose $N$ so that the contribution of the $\theta$ level truncation to the overall error is controlled by the $\epsilon_\theta$, following the same idea as in \cite{ishwaran2001gibbs}. Then, for each $k$, we compute a preliminary value $\tilde{M}_k$ that would make the $\psi$ level contribution satisfy the remaining error budget $\epsilon-\epsilon_\theta$. Let $k^*$ denote the index corresponding to the smallest value of $(\tilde{M}_k-1)/\alpha^{k}_{\psi\mid\theta}$. We set $M_{k^*}=\tilde{M}_{k^*}$, and choose the remaining $M_k$'s so that $(M_k-1)/\alpha^{k}_{\psi\mid\theta}$ is approximately matched across clusters. In this way, the truncation error contributed by each $\psi$ subcluster is kept at a comparable level.

\begin{Remark}

The calculation of $N$ and $M_k$'s depends on the values of the $\alpha$'s. We propose to use VB to estimate the values of $\alpha$ here. We assess this approach in Sections $\ref{sec:sim}$ and $\ref{sec:application}$.
\end{Remark}

\subsection{Comparison with the constant truncation bound}

We now compare the truncation with $M_K=M$ from \cite{burns2023truncation} vs $M_k\neq M$. We note that the Theorem $\ref{thm:inequality}$ holds in both scenarios. Hence, for $M_k\neq M$ case, the RHS of $(\ref{eq:ineqality_thm})$ satisfies  

\begin{equation}\label{eq:comp1}
     n\exp\Big(-\min_k\Big\{\frac{M_k-1}{\alpha_k^{\psi|\theta}}\Big\}\Big)\Big)\Big(1-\dfrac{\epsilon_\theta}{n}\Big)\leq\epsilon-\epsilon_\theta.
\end{equation}

\noindent For $M_k=M$, we have 

\begin{equation}\label{eq:comp2}
     n\exp\Big(-\Big\{\frac{M-1}{\alpha^{\psi|\theta}}\Big\}\Big)\Big)\Big(1-\dfrac{\epsilon_\theta}{n}\Big)\leq\epsilon-\epsilon_\theta.
\end{equation}
From $(\ref{eq:cluster_values})$, $(\ref{eq:comp2})$is equivalent to $(\ref{eq:comp1})$, only for $k=k^*$. Also from $(\ref{eq:cluster_values})$, for any $k$, $$M_k=\ceil{\Big\{1+\Big(\dfrac{\alpha_k^{\psi|\theta}}{\alpha_{k^*}^{\psi|\theta}}\Big)(M_{k^*}-1)\Big\}}\hspace{1cm}\forall k\neq k^*.$$ So we have
$$\Big\{1+\dfrac{\alpha_k^{\psi|\theta}}{\alpha_{k^*}^{\psi|\theta}}(M_{k^*}-1)\Big\}<M_k<\Big\{1+\Big(\dfrac{\alpha_k^{\psi|\theta}}{\alpha_{k^*}^{\psi|\theta}}\Big)(M_{k^*}-1)\Big\}+1\hspace{1cm}\forall k\neq k^*$$

$$\frac{M_{k^*}-1}{\alpha^{\psi|\theta}_{k^*}}\leq\frac{M_k-1}{\alpha^{\psi|\theta}_k}<\frac{1}{\alpha^{\psi|\theta}_k}+\frac{M_{k^*}-1}{\alpha^{\psi|\theta}_{k^*}}$$ where $k^*$ corresponds to  $$\dfrac{M_k^*-1}{\alpha_{k^*}^{\psi|\theta}}=\arg\min_k \dfrac{\tilde{M}_k-1}{\alpha_k^{\psi|\theta}}.$$ However, since $M_k^*$ also satisfies $(\ref{eq:comp2})$, by the specification of $M$, we can conclude $M_k\leq M_k^*\leq M$, i.e. by assuming variable cluster size, we can reduce (and tighten) the number of total clusters.

Table $\ref{tab:N-M table}$ shows the values $N$ and $M_k$'s using $(\ref{eq:cluster_values})$ for different $\alpha$ values $(\alpha^\theta \text{ and }\alpha^{\psi|\theta})$ with the error level for the $\theta$ clusters kept at $0.001$ and the overall error at $0.01$. Using values of $N$ and $M_k$'s derived from $(\ref{eq:cluster_values})$, we see that increasing $\alpha^\theta$ increases the number of $\theta$ clusters. And increasing the ratio of $\dfrac{\alpha^{\psi|\theta}}{\alpha^\theta}$ increases the number of $\psi$ clusters.

 \begin{table}[h!]
     \centering
     \begin{tabular}{|c|c|c|c|}
     \hline
         $\alpha^\theta$ & $\alpha^{\psi|\theta}$ & $n=200$ & $n=1000$ \\
         \hline
         $0.5$ & $(0.5,0.5,0.5, \ldots)$ & $N=8, M=(6,6,6, \ldots)$ & $N=8, M=(7,7,7, \ldots)$\\
         $0.5$ & $(0.5,1,1.5, \ldots)$ & $N=8, M=(6,11,16, \ldots)$ & $N=8, M=(7,13,19, \ldots)$\\
         $1$ & $(0.5,1.5,3,\ldots)$ & $N=14, M=(6,16,31, \ldots)$ & $N=15, M=(7,19,36, \ldots)$\\
         $3$ & $(0.5,1.5,3, \ldots)$ & $N=38, M=(6,16,31, \ldots)$ & $N=36, M=(7,19,36, \ldots)$\\
         \hline
     \end{tabular}
     \caption{Truncation values of $N$ and $M$ corresponding to an error of $0.001$ for $\theta$ clusters and an overall error of $0.01$}
     \label{tab:N-M table}
 \end{table}

One of the key advantage of using this finer truncation approximation is the computational efficiency. Varying the truncation bounds, the computation complexity is proportional to $O(\sum_{i=1}^N M_k)$; and for fixed truncation bounds as in \cite{burns2023truncation}, it is proportional to $O(NM)$. Since $M_k\leq M\hspace{0.1cm}\forall k$, computational complexity is reduced using varying truncation bounds.

\subsection{Blocked Gibbs Sampler}

The truncation approximation allows for a Blocked Gibbs Sampler (BGS) to work. The specification of $Y$ and $X$ models determine the form of the full conditionals for their parameters. We provide details in the simulation in Section $\ref{sec:sim}$. However the conditionals of $V$ and $\alpha^\theta$, $\alpha^{\psi|\theta}$ have the same form for all the EDPM specifications and are given below.

    $$V^\theta_k|\alpha^\theta \stackrel{ind}{\sim}Beta(n_k+1,\alpha^\theta + \sum_{i=k+1}^N n_i),$$
    where $n_k$ is the number of observation currently assigned to the $k^{th}$ $\theta$ cluster.
      $$V^\psi_{j|k}|\alpha^{\psi|\theta}_k \stackrel{ind}{\sim}Beta(n_{kj}+1,\alpha^{\psi|\theta}_k + \sum_{i=j+1}^{M_k} n_{ki}),$$
    where $n_{kj}$ is the number of observations currently assigned to the $j^{th}$ $\psi$ sub-cluster of the $k^{th}$ $\theta$ cluster.
     $$\alpha^\theta \sim Gamma(N-1,-\sum_{k=1}^N\log(1-V_k^\theta))$$
     $$\alpha_k^{\psi|\theta} \sim Gamma(M_k-1,-\sum_{j=1}^{M_k}\log(1-V_{j|k}^\psi)).$$

For details on our specification for the simulation example, see Section $\ref{sec:gibbs_dist_spec}$.

\section{Simulation to assess the mixing of the BGS based on the truncation approximation}\label{sec:sim}
We assess the mixing of the BGS algorithm for arbitrarily large $N,M_k$; $N,M_k$ chosen as in \cite{burns2023truncation} but using VB estimates of $\alpha$'s and $N,M_k$ computed from $(\ref{eq:cluster_values})$ again using VB estimates for the $\alpha$'s. We compare the methods by assessing the mixing of the BGS algorithm for estimating $E(Y|X)$ using batch means from \cite{burns2023truncation}. We simulate $100$ replicated datasets with $n=200$. We vary the dimension of $X$ as $5,10,15$ to see the effect of dimension.

We will generate the data (from an EDPM) as follows:

\begin{equation}\label{Eq:EDPM_setup_sim}
\begin{split}
& V_k^\theta \stackrel{iid}{\sim} Beta (1,\alpha_\theta)\hspace{0.5cm}1\leq k\leq N(=10)-1; \hspace{0.2cm} \text{Set } V_0^\theta = 0,\hspace{0.1cm}V_N^\theta = 1,\\
& V_k^\psi \stackrel{iid}{\sim} Beta (1,\alpha_\psi)\hspace{0.5cm}1\leq k\leq M(=6)-1; \hspace{0.2cm} \text{Set } V_0^\psi = 0, \hspace{0.1cm}V_M^\psi = 1\\
& \theta\sim N_n(\mu_\theta,\sigma_\theta^2I_n)\\
& \psi\sim N_m(\mu_\psi,\sigma_\psi^2I_n)\\
& p_k^{\theta} = V_k^{\theta}\prod_{j=0}^{k-1}(1-V_j^{\theta})\hspace{0.5cm}1\leq k\leq N \hspace{0.2cm}(\text{Then} \sum_{k=1}^N p_k^\theta=1)\\
& p_k^{\psi} = V_k^{\psi}\prod_{j=0}^{k-1}(1-V_j^{\psi})\hspace{0.5cm}1\leq k\leq M_i(=M) \hspace{0.2cm}(\text{Then} \sum_{k=1}^{M_i} p_k^\theta=1)\\
& z \sim Multinomial (n,p^\theta)\\
& z_{.|i} \sim Multinomial (m,p^\psi)\\
& X_{ij}\stackrel{iid}{\sim} N(\psi_{z_{i|i}},\sigma^2)\\
& X_i = (X_{i1}, X_{i2}, \ldots X_{im})\\
& Y_i|X_i, \theta_i \sim f(y|x,\theta)\\
\end{split}    
\end{equation}

\noindent $(\theta_z,\psi_{.|.})$ follows a EDP distribution with parameters $(\alpha_\theta,\alpha_\psi,p^\theta \times p^\psi, \mu_\theta, \sigma_\theta,\mu_\psi, \sigma_\psi)$.

\noindent We sampled $\alpha_\theta \sim \Gamma(1,1)$ and $\alpha_\psi \sim \Gamma(1,1)$and set $\sigma=\frac{1}{4}$ and $n=200$ and considered $dim(X)\in \{5,10,15\}$. For the simplicity, we set $M_k=M$ $\forall k$.

\vspace{0.2cm}
    
\noindent At each MCMC iteration, $E(Y|X)$ has the following form:

\begin{equation}\label{eq:expec}
\begin{split}
    E(Y|X) & =E(E(Y|X,\theta_{z_1},\ldots \theta_{z_N}))\\
    & = \sum_{k=1}^N \sum_{j=1}^{M_k} \frac{p^\theta_ip^\psi_{j|i}E(Y_i|X_i,\theta_{z_i})f(X_i|\psi_{z_{i|j}})}{\sum_i\sum_j p^\theta_ip^\psi_{j|i}f(X_i|\psi_{z_{i|j})}}.\\
\end{split}    
\end{equation}
    
\noindent We consider two different specifications of $f(y|x,\theta)$ in $(\ref{Eq:EDPM_setup_sim})$.

We collect $400$ batches of sampled values, each of size $100$. For each batch we computed the $0.25$, $0.75$ quantile and the mean. We then computed the mean value and standard deviation over the $400$ batches for each of the summary statistics. The averaged value over $i$ (over the batches) is the (averaged) mean value and (averaged) standard deviation value we report here. A smaller SD signifies better mixing (\cite{ishwaran2001gibbs}).

\noindent The models within the EDPM, based on which we will derive the BGS posterior distributions, are as follows:

\begin{equation*}
\begin{split}
    Y_i|X_i,\theta_i \sim N(X_i^T\mu_{\theta_i},\sigma_{\theta_i}^2)\\
    x_{il}|\psi_{il} \sim N(\mu_{\psi_{il}}, \sigma_{\psi_{il}}^2)~\forall l\in \{1,\ldots p\}
\end{split}    
\end{equation*}

\subsection{Scenario I}

For the first scenario, we use the following form for $f(y|x,\theta)$ of $(\ref{Eq:EDPM_setup_sim})$ in the EDPM, 

\begin{equation*}\label{eq:setup}
\begin{split}
& Y_i|X_i \sim \lambda_{X_i} N(X_i\theta_z, \sigma^2)+(1-\lambda_{X_i})t_{2X_i\theta_z}\\
& \lambda_{X_i}=\dfrac{\omega_1\exp\{-\frac{\omega_1}{2}(X_{i,1}-\mu_1)^2\}}{\omega_1\exp\{-\frac{\omega_1}{2}(X_{i,1}-\mu_1)^2\}+\omega_2\exp\{-\frac{\omega_2}{2}(X_{i,1}-\mu_2)^2\}}\\
&\omega_1=\omega_2=2,\\
\end{split}
\end{equation*}

\noindent where $X,\theta,z$ are generated using $(\ref{Eq:EDPM_setup_sim})$. $t_{2X_i\theta_z}$ is a $t$ distribution with df = $2X_i\theta_z$, the true model will be a mixture of this complicated mixture distribution. Note that the mixture proportion depends on only the first covariate.

We quantify the mixing of $E(Y|X)$ of $(\ref{eq:expec})$, for $X$ generated from $(\ref{Eq:EDPM_setup_sim})$ for different dimensions in Table $\ref{tab:tab_100}$. 

The batch means for all three cases are the same. However, we generally observe smaller batch SD for $N$ and $M_k$ using $(\ref{eq:cluster_values})$ Thus the mixing appears better and the computational time is faster than using either large $N$ and $M_k$ or $M_k=M$.

\begin{table}[]
    \centering
    \begin{tabular}{|c|c|c|c|c|c|c|c|}
    \hline
    & & \multicolumn{2}{|c|}{$N$,$M_k$ from $(\ref{eq:cluster_values})$} & \multicolumn{2}{|c|}{Large $N$,$M_k$ } & \multicolumn{2}{|c|}{$M_k=M$}\\ \cline{3-8}
    \hline
    dim of $X$ & Statistic & Mean & SD & Mean & SD & mean & SD\\
    \hline
     5 & $0.25$ quantile & 2.25 & 0.020 & 2.25 & 0.021 & 2.25 & 0.026\\
     & mean & 2.40 & 0.015 & 2.40 & 0.017 & 2.40 & 0.021\\
     & $0.75$ quantile & 2.52 & 0.023 & 2.52 & 0.021 & 2.52 & 0.024\\
     \hline
     10 & $0.25$ quantile & 2.39 & 0.013 & 2.39 & 0.016 & 2.39 & 0.018 \\
     & mean & 2.56 & 0.013 & 2.56 & 0.014 & 2.56 & 0.019\\
     & $0.75$ quantile & 2.81 & 0.019 & 2.81 & 0.018 & 2.81 & 0.020\\
     \hline
     15 & $0.25$ quantile & 2.49 & 0.007 & 2.49 & 0.009 & 2.49 & 0.012\\
     & mean & 2.64 & 0.010 & 2.64 & 0.009 & 2.64 & 0.011\\
     & $0.75$ quantile & 2.89 & 0.010 & 2.89 & 0.013 & 2.89 & 0.012\\
     \hline
    \end{tabular}
    \caption{Mean value and SD for $400$ batches with $100$ iterations to calculate $3$ different quantiles for the distribution of $E(Y|X)$. The mean and SD is computed over $100$ replicated datasets for scenario I.}
    \label{tab:tab_100}
\end{table}

\newpage
\subsection{Scenario II}

For the second scenario, we assume the form of $f(y|x,\theta)$ in $(\ref{Eq:EDPM_setup_sim})$ to be

\begin{equation*}\label{eq:setup2}
\begin{split}
& Y_i|X_i \sim N(X_i\theta_z, \sigma^2);\\
\end{split}
\end{equation*}

\begin{table}[]
    \centering
    \begin{tabular}{|c|c|c|c|c|c|c|c|}
    \hline
    & & \multicolumn{2}{|c|}{$N,$ $M_k$ from $(\ref{eq:cluster_values})$} & \multicolumn{2}{|c|}{Large $N,$ $M_k$} & \multicolumn{2}{|c|}{$M_k=M$}\\ \cline{3-8}
    \hline
    dim of $X$ & Statistic & Mean & SD & Mean & SD & Mean & SD\\
    \hline
    5 & $0.25$ quantile & 1.85 & 0.014 & 1.85 & 0.016 & 1.85 & 0.018\\
     & mean & 1.97 & 0.013 & 1.98 & 0.016 & 1.97 & 0.019\\
     & $0.75$ quantile & 2.16 & 0.019 & 2.13 & 0.017 & 2.16 & 0.021\\
     \hline
     10 & $0.25$ quantile & 1.91 & 0.011 & 1.91 & 0.013 & 1.91 & 0.016\\
     & mean & 2.01 & 0.007 & 2.01 & 0.010 & 2.01 & 0.011\\
     & $0.75$ quantile & 2.19 & 0.015 & 2.19 & 0.013 & 2.19 & 0.018\\
     \hline
     15 & $0.25$ quantile & 1.92 & 0.008 & 1.92 & 0.008 & 1.92 & 0.010\\
     & mean & 2.03 & 0.008 & 2.03 & 0.009 & 2.03 & 0.012\\
     & $0.75$ quantile & 2.20 & 0.012 & 2.20 & 0.014 & 2.20 & 0.015\\
     \hline
    \end{tabular}
    \caption{Mean value and SD for $400$ batches with $100$ iterations to calculate $3$ different quantiles for the distribution of $E(Y|X)$. The mean and SD is computed over $100$ replicated datasets for scenario II.}
    \label{tab:tab_100_true_nomix}
\end{table}

\noindent The results can be found in the Table $\ref{tab:tab_100_true_nomix}$. Similar to the first scenario, we see the equality in estimation of the batch means for all the choices of $N$ and $M_k$; however, lower (though the difference is less than the first setup) SD for $N$ and $M_k$ chosen from $(\ref{eq:cluster_values})$ for almost all quantiles indicates the MCMC algorithm mixes better than using arbitrarily large $N$ and $M_k$.

\clearpage

\subsection{Posterior distributions for Block Gibbs Sampler}\label{sec:gibbs_dist_spec}

The full conditional distributions for the parameters of $Y|X$ and $X$  models for the EDPM specifications in the simulations and the data example are given below.

     $$\mu_\theta\sim N((X^TX+C_y)^{-1}(X^TY+C_y\mu_0),\sigma_\theta^2(X^TX+C_y)^{-1})$$
    where $\mu_0$ is the initial value for $\mu$.
    $$\mu_{\psi_{j|k},l}\sim N((\frac{n_{kj}}{\sigma_\psi^2}+\frac{1}{c_{x,l}})^{-1}(\frac{\sum_{i|K_i=k,J_i=j}x_{il}}{\sigma_\psi^2}+\frac{m_l}{c_{x,l}}),(\frac{n_{kj}}{\sigma_\psi^2}+\frac{1}{c_{x,l}}^{-1})$$ 
    where $K$ and $J$ are the cluster assignment vectors.
     $$\dfrac{1}{\sigma^2_\theta}\sim Gamma \Big(\frac{N}{2},\sum_{i}\frac{(y_i-\mu_\theta)^2}{2}\Big)$$
     $$\dfrac{1}{\sigma_\psi^2}\sim Gamma \Big(\frac{M_k}{2},\sum_{i}\frac{(y_i-\mu_{\psi_{j|k},l})^2}{2}\Big)$$
     And finally for the distribution of the cluster parameters $K$ and $J$, For each observation $i$ and a $\theta,\psi$ cluster choice, the probability the observation $i$ is in the $\theta$ cluster $k$ and $\psi$ cluster $j$ is $$\dfrac{p_k^\theta p_{j|k}^\psi}{\sqrt{\sigma_\theta \sigma_\psi^{M_k}}}\exp{-\dfrac{1}{2}\Big[\dfrac{(y_i-(X\mu_\theta)_i)^2}{\sigma_\theta^2}+\sum_{j=1}^{M_k}\dfrac{(x_{il}-\mu_{\psi_{j|k},l})^2}{\sigma_\psi^2}\Big]}$$

\section{Application}\label{sec:application}

The Third National Health and Nutrition Examination Survey (NHANES III) is a series of national examination studies conducted by the National Center for Health Statistics, Centers for Disease Control and Prevention during $1988$ to $1994$.\nocite{NHIII} We will work with the phase 1 data of this survey conducted between $1988$ and $1991$. Our goal is to fit a model to the number of years for a person to reach a fatal CVD event without any prior history. For illustration, we focus only on deaths (and remove the censored).

The predictor variables are the current age of subject at the time of survey(AGE), c-reactive protein in blood (CRP), race ($1=$White, $2= $Black), sex ($1=$Men, $2=$ Women) of subjects, BMI, waist hip ratio, blood pressure, (total and good) cholesterol, and glucose.

\subsection{Results from parametric AFT}

\begin{figure}[h!]
    \centering
    \includegraphics[width=0.9\linewidth]{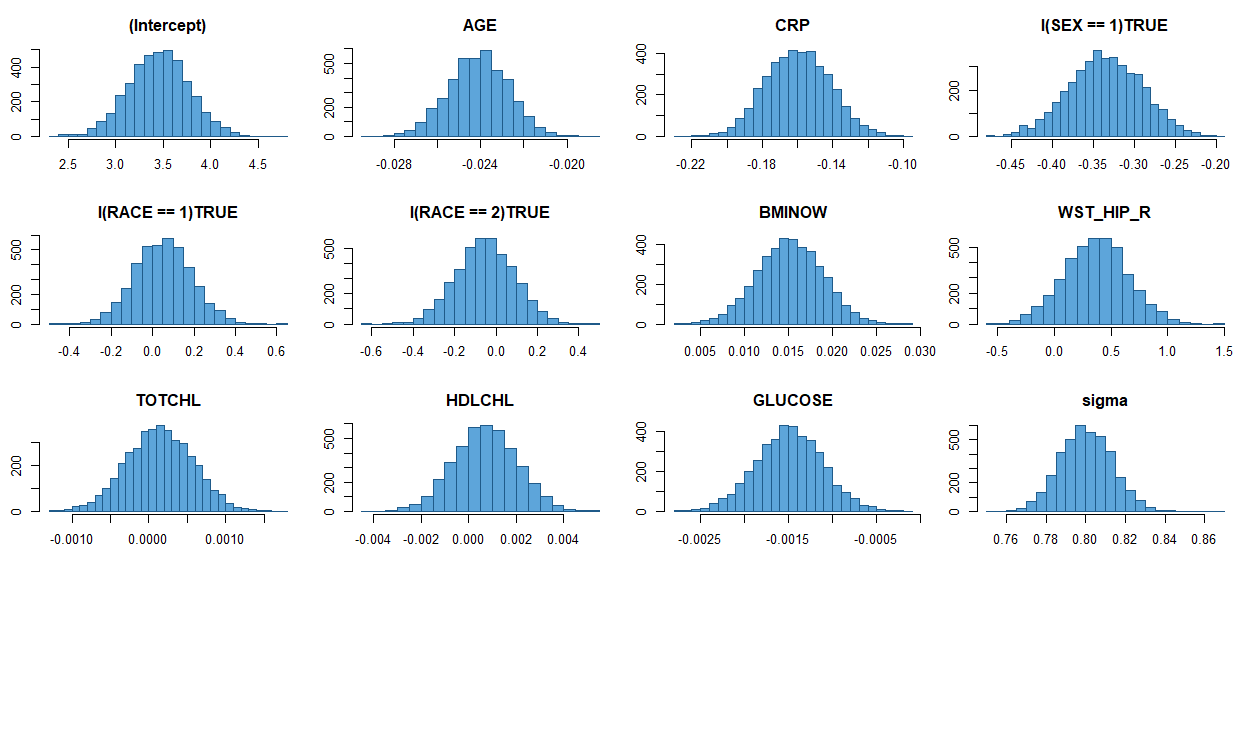}
    \caption{Posterior for regression coefficients of each covariates.}
    \label{fig:enter-label}
\end{figure}

\noindent First we fit a log-normal AFT model. Figure $\ref{fig:enter-label}$ shows the histogram of the posterior distribution of coefficient for each of the variables. Note that for covariates race, waist hip ratio, total cholesterol and HDL, the $95\%$ credible interval (CI) contains $0$. Hence, the time to CVD death is not statistically significant for these variables. However, for age, c-reactive protein, sex, BMI and blood glucose level, the $95\%$ CI does not contain $0$. Hence, we can infer these variables are statistically significant.

\noindent Next we will fit the data under three different EDPM truncations and compare to parametric AFT.

\subsection{Results for Truncated EDPM}

First we fit the EDPM using the truncation approach proposed here. To do this we need the estimated $\alpha$'s. Using VB, the estimated $\alpha$'s are $\alpha_\theta=0.09$, and $\alpha_\psi=(0.12,0.22)$.  Based on these $\alpha$ values, we set the truncation values $N=2$ Y-clusters and $M=(3,5)$ as X clusters. We will compare the parametric AFT and various truncation approximations based on the posterior mean of the rank correlations between observed and predicted.

\begin{figure}[h!]
    \centering
    \includegraphics[width=0.9\linewidth]{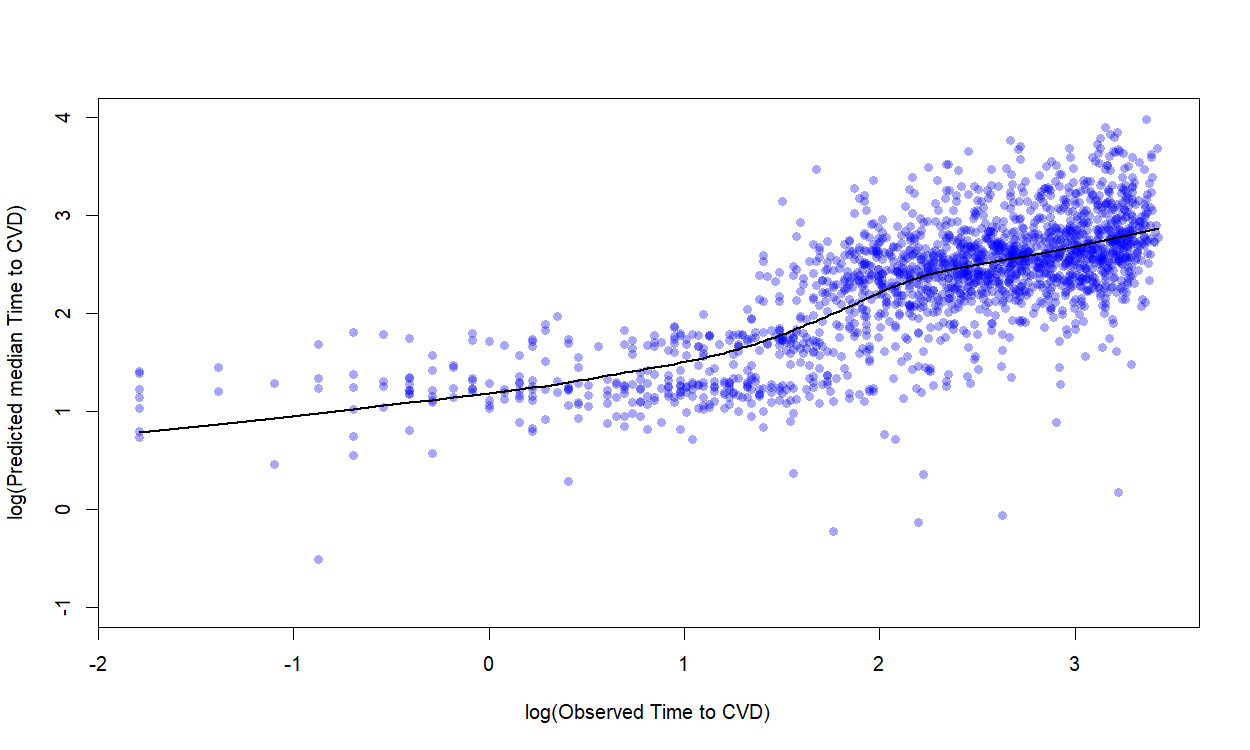}
    \caption{EDPM predicted vs observed time to CVD death on the logscale with truncation with varying number of  X-clusters in each Y-cluster and with $\alpha$'s chosen using VB.}
    \label{fig:EDPM}
\end{figure}

\noindent Figure $\ref{fig:EDPM}$ shows the EDPM fit of the data via a plot of observed vs predicted time to CVD death. The mean rank correlation between observed and predicted values is $0.75$ with a posterior standard deviation of $0.014$.

\begin{figure}[h!]
    \centering
    \includegraphics[width=0.9\linewidth]{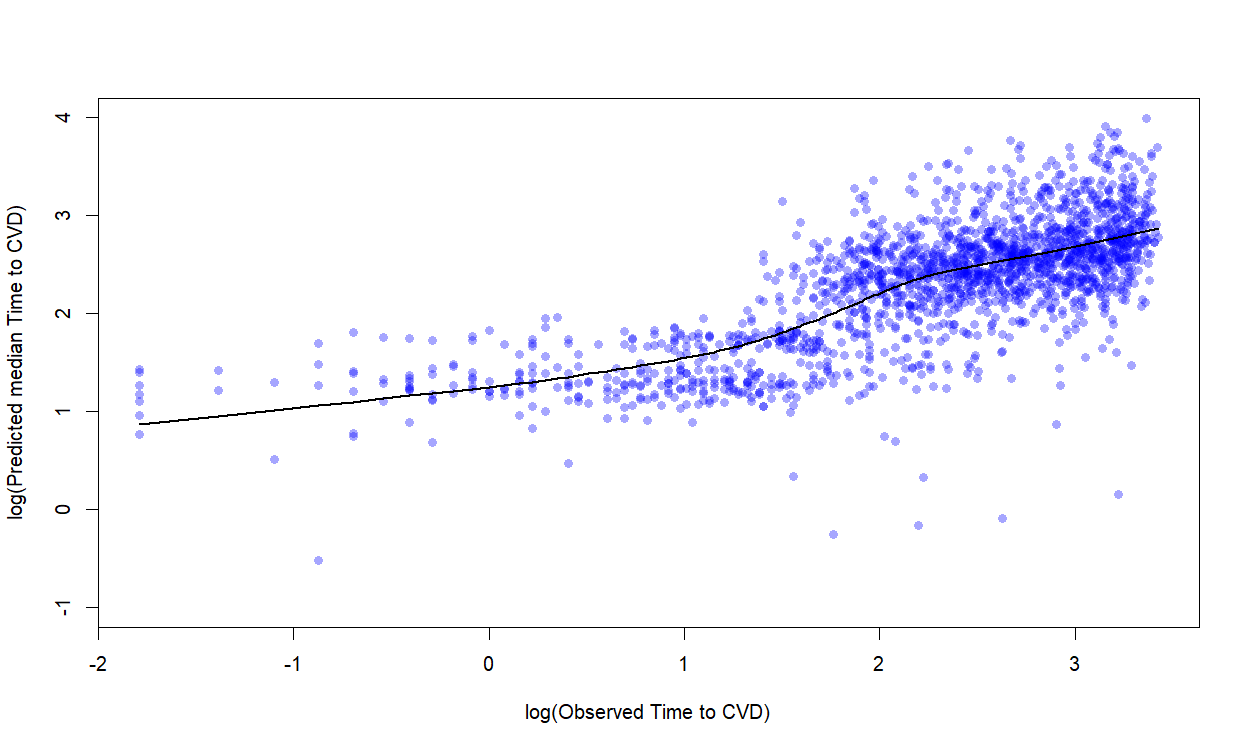}
    \caption{EDPM predicted vs observed time to CVD death on the logscale assuming $M_k=M$ $\forall$ $k$ and $\alpha$'s chosen by VB.}
    \label{fig:EDPM_natalie}
\end{figure}

\clearpage

Figure $\ref{fig:EDPM_natalie}$ shows the EDPM fit of the data via a plot of observed vs predicted time to CVD death with truncation values recommended in \cite{burns2023truncation}, with  $N=2$, $M=(5,5)$. The mean rank correlation is $0.73$ with a posterior standard deviation of $0.017$. The correlation is similar but lower than result in Figure \ref{fig:EDPM}.

For large truncation values, $N=5$, $M=(10,10,10,10,10)$, the mean rank correlation between observed and predicted values is $0.72$ with a posterior standard deviation of $0.015$. As expected, the computation time in this case is significantly higher than the previous one, but the correlation is very similar. We note that the mixing is worse with unnecessarily large truncation values.

\begin{figure}[h!]
    \centering
    \includegraphics[width=0.9\linewidth]{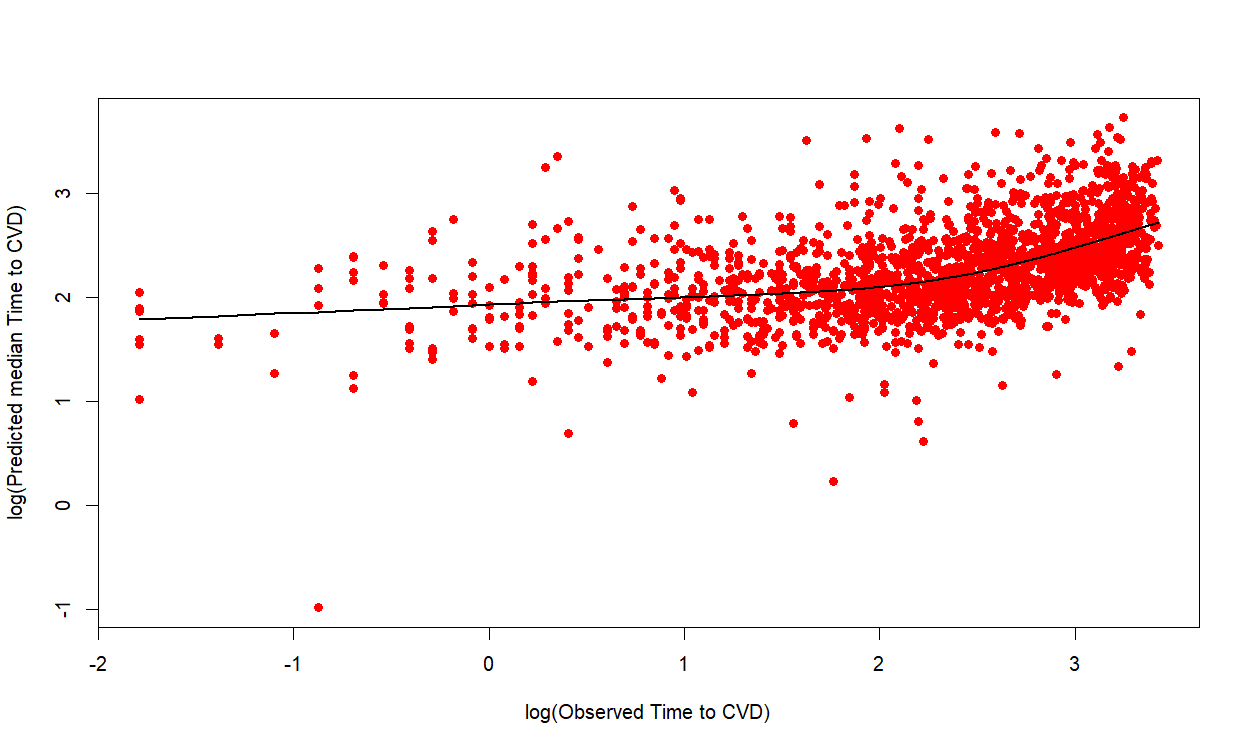}
    \caption{Lognormal AFT predicted vs observed time to CVD death on the logscale.}
    \label{fig:AFT logscale}
\end{figure}

\noindent Figure $\ref{fig:AFT logscale}$ shows the parametric lognormal AFT fit. The mean rank correlation is $0.55$ with a standard deviation of $0.019$.

\noindent From the Figures $\ref{fig:EDPM}$ - $\ref{fig:AFT logscale}$ and the posterior estimates of the rank correlation and standard deviations, all the truncated EDPMs outperform the parametric as expected. The finer truncation proposed here performs essentially the same as the other truncations but is more computationally efficient.

\clearpage

\section{Discussion}\label{sec:conclusion}

In this work, we develop a new EDPM truncation approximation and a Variational Bayes (VB) algorithm for a truncated EDPM. We first show the convergence of the truncated EDPM to the untruncated and how to assess the accuracy of the approximation. Then we calculate the bounds for the truncation using VB, given the desired error level. We develop a Block Gibbs Sampler (BGS) for posterior inference. We did a simulation study to assess the mixing of our truncation approximation. We concluded that choosing the truncation level based on VB mixes better than arbitrarily large truncation values. We also noted that using VB to calculate the truncation results in much faster computations than assuming arbitrarily large truncation values.

We applied the approach to the data from the third NHANES.  We compared the Enriched Dirichlet Process Mixture (EDPM) model with truncation values calculated using VB with an EDPM with arbitrarily large truncation value, EDPM with truncation values from \cite{burns2023truncation} using VB, and a classical log-normal AFT model. We concluded all the EDPMs performed better than the parametric AFT model. The three EDPM models have similar performances, but the EDPM with truncations values that vary by Y-cluster and with $\alpha$'s calculated using VB is computationally most efficient.

For future work, we can extend the EDPMs to time-varying co-variates, which will broaden its applicability in survival analysis. We can also explore assessing the accuracy of the truncation approximation in DDP-GPs from \cite{Xu_2016}.

\section*{Acknowledgement}

Bhadra and Daniels were partially supported by NIH R$01$ $166324$.

\newpage
\bibliographystyle{plainnat}
\bibliography{people}

\end{document}